\documentclass[12pt]{article}
\usepackage{amssymb}
\usepackage{amsmath}
\usepackage{amsthm}
\usepackage{setspace}
\usepackage{array}
\usepackage{times}
\usepackage{float}
\usepackage{bm}
\usepackage{algorithm}
\usepackage{tcolorbox} 

\usepackage[colorlinks=true,citecolor=blue,urlcolor=blue,linkcolor=blue]{hyperref}
\usepackage[round]{natbib}

\newcommand{\RN}[1]{%
  \textup{\uppercase\expandafter{\romannumeral#1}}%
}
\usepackage{tikz}
\usetikzlibrary{shapes,arrows,arrows.meta}
\tikzstyle{decision} = [diamond, draw, text width=1em, text badly centered, node distance=3cm, inner sep=0pt]
\tikzstyle{block} = [circle, draw, text width=1em, text centered, minimum height=2em]

\usepackage{graphicx}
\usepackage{enumerate}
\usepackage{latexsym}
\usepackage{verbatim}
\usepackage{color}
\usepackage{longtable}
\usepackage{booktabs}
\usepackage{subcaption}

\usepackage{cite}
\usepackage[dvips, letterpaper, nohead, top=1.3in, bottom=1.3in, left=1.3in, right=1.3in]{geometry}

\newtheorem{theorem}{Theorem}[section]

\newtheorem{definition}{Definition}[section]
\newtheorem{procedure}{Procedure}[section]
\newtheorem{proposition}{Proposition}[section]

\newtheorem{remark}{Remark}[section]
\newtheorem{example}{Example}[section]

\normalsize

\usepackage{xr}
\usepackage{authblk}

\title{Weighted Holm Procedures: Theory, Properties, and Recommendations}

\author[1]{Beibei Li}
\author[2]{Wenge Guo\thanks{Author e-mail addresses: beibei.li@abbvie.com; wenge.guo@njit.edu}}

\affil[1]{Data and Statistical Sciences, AbbVie Inc., Florham Park, NJ, USA}
\affil[2]{Department of Mathematical Sciences, New Jersey Institute of Technology, Newark, NJ, USA}

\date{\today}

\begin{document}
\maketitle

\begin{abstract}
\noindent In many statistical applications, particularly in clinical studies, hypotheses may carry different levels of importance, motivating the use of weighted multiple testing procedures (wMTPs) to control the familywise error rate (FWER). Among these approaches, two weighted Holm procedures are commonly used: the weighted Holm procedure (WHP), which is based on ordered weighted $p$-values, and the weighted alternative Holm procedure (WAP), which relies on ordered raw $p$-values. 
This paper provides a systematic comparison of these two procedures, along with practical recommendations for their use. We first examine their corresponding closed testing procedures (CTPs) and show that WHP is uniformly more powerful than WAP. We further investigate their structural properties, demonstrating that WAP, while consonant, lacks monotonicity. 
To facilitate communication with non-statisticians, we introduce graphical representations of both procedures using a common initial graph and distinct updating strategies. In addition, we derive adjusted $p$-values and adjusted weighted $p$-values for both methods. Finally, we establish an optimality result: WHP cannot be improved by enlarging any of its critical values without violating FWER control, whereas WAP is optimal only under specific conditions. Simulation studies support these theoretical findings and highlight the superior FWER control and average power of WHP.
\end{abstract}


	
\section{Introduction}	

\noindent In many clinical applications, hypotheses differ in importance, this motivates the use of weighted multiple testing procedures {(wMTPs). Several such methods have been developed, including weighted Bonferroni, Simes, and resampling-based approaches \citep{rosenthal1983ensemble, holm1979simple, benjamini1997multiple, tamhane2008weighted, westfall1993resampling}. A weighted Bonferroni procedure, for example, can increase power for more important hypotheses. Among stepwise methods, two weighted Holm procedures are widely used: the weighted Holm procedure (WHP; \citeauthor{holm1979simple}, \citeyear{holm1979simple}), which is based on ordered weighted $p$-values, and the weighted alternative Holm procedure (WAP; \citeauthor{benjamini1997multiple}, \citeyear{benjamini1997multiple}), which relies on ordered raw $p$-values. WAP is often viewed as more objective since its ordering is independent of weights, but this comes at the cost of losing monotonicity, a key property in stepwise testing. Prior studies have examined weighted procedures under pre-ordered hypotheses \citep{maurer1995multiple, wiens2003fixed, wiens2013selection} and weighted Simes-type approaches \citep{tamhane2008weighted}, but the relative merits of ordering by raw versus weighted $p$-values remain unresolved. This motivates a careful comparison of WHP and WAP to clarify their statistical properties, evaluate their performance in terms of familywise error rate (FWER) control and power, and provide practical guidance. We show that WHP is uniformly more powerful than WAP, establish new optimality results, and demonstrate the advantages of WHP through theory and simulation.

\smallskip
Bonferroni-type weighted procedures can also be formulated as closed testing procedures (CTPs), which offer a rigorous framework for FWER control \citep{marcus1976closed, brannath2010shortcuts, henning2015closed}. However, as the number of hypotheses grows, the number of intersection tests increases rapidly, making CTPs difficult to apply in practice. Graphical methods provide a more accessible alternative: they can represent a wide range of multiple testing procedures (MTPs), are easier to communicate to clinical teams than lengthy decision tables, and reduce programming effort. They are now widely used in clinical research \citep{bretz2009graphical, burman2009recycling, bretz2011graphical}. Weighted procedures also underpin key gatekeeping strategies, including serial, parallel, and mixture approaches 
\citep{maurer1995multiple, westfall2001optimally, dmitrienko2003gatekeeping, dmitrienko2007gatekeeping, dmitrienko2013general}. Developing simple and powerful weighted procedures is therefore essential for advancing these frameworks. Moreover, their optimality—whether they can be improved without sacrificing error control—remains largely unexplored, especially for WHP and WAP.

\smallskip
Although this paper does not address the choice of weights, their importance should be emphasized. Weights may be assigned based on the a priori importance of hypotheses or prior knowledge \citep{westfall2001optimally}, or chosen adaptively from data to improve power while preserving significance levels. Several approaches have been proposed for data-dependent weighting that still maintain familywise or generalized familywise error control \citep{finos2007fdr, kang2009weighted, dalmasso2008weighted, westfall2004weighted, wang2019weighted}. In addition, weighted parametric procedures that exploit the joint distribution of test statistics have been studied \citep{xie2012weighted, xi2017unified}.

\smallskip
In summary, we compare WHP and WAP through their closed testing formulations, graphical representations, adjusted $p$-values, and optimality, and provide practical recommendations. We establish that WHP is uniformly more powerful than WAP. Although WAP lacks monotonicity—both with respect to raw $p$-values and within the closed testing framework—it remains consonant, which enables computational shortcuts, and its graphical formulation extends beyond the settings of \citet{hommel2007powerful} and \citet{bretz2009graphical}. Regarding optimality, WHP cannot be improved without sacrificing FWER control, whereas WAP is optimal only under specific weight-ratio conditions. Simulation studies and clinical examples further demonstrate that WHP consistently outperforms WAP, and that judiciously chosen weights enhance the power of both procedures, particularly WAP.  

The remainder of the paper is organized as follows. Section~\ref{prem} introduces the notation. Section~\ref{ctps} presents the CTPs and associated monotonicity results. Section~\ref{graphs} develops graphical representations. Section~\ref{adjs} derives adjusted $p$-values. Section~\ref{optimality} establishes optimality properties. Section~\ref{simul} presents simulation results. Section~\ref{realdata} illustrates applications with clinical examples. Section~\ref{summary2} concludes. All proofs are collected in Section~\ref{append}.

\section{Preliminaries}\label{prem}
    
In this section, we introduce notation and describe the two weighted Holm procedures studied in this paper.  

\subsection{Basic Notation}
We consider the simultaneous testing of $m$ hypotheses $H_1,\ldots,H_m$ with corresponding $p$-values $P_1,\ldots,P_m$ and positive weights $w_1,\ldots,w_m$. The weighted $p$-values are defined as  
\[
\tilde{P}_i = \frac{P_i}{w_i}, \quad i=1,\ldots,m.
\]  

Let $P_{(1)} \leq \cdots \leq P_{(m)}$ denote the ordered raw $p$-values with associated hypotheses $H_{(1)},\ldots,H_{(m)}$ and weights $w_{(1)},\ldots,w_{(m)}$. Similarly, let $\tilde{P}_{(1)} \leq \cdots \leq \tilde{P}_{(m)}$ denote the ordered weighted $p$-values, with corresponding hypotheses $H^{*}_{(1)},\ldots,H^{*}_{(m)}$ and weights $w^{*}_{(1)},\ldots,w^{*}_{(m)}$.  

Suppose that $m_0$ of the hypotheses are true nulls and $m_1 = m - m_0$ are false nulls, and let $I_0$ be the index set of the true null hypotheses. Denote by $R$ the total number of rejections and by $V$ the number of true nulls rejected. The FWER is defined as  
\[
\text{FWER} = \Pr(V \geq 1),
\]  
i.e., the probability of making at least one type~I error \citep{tamhane2018advances}. Our goal is to control the FWER at a pre-specified significance level $\alpha \in (0,1)$.  

Throughout, we assume that for each $i \in I_0$, the corresponding $p$-value satisfies the standard super-uniform condition  
\[
\Pr(P_i \leq x) \leq x, \quad \forall\, x \in (0, 1),
\]  
without imposing any restrictions on their joint distribution.

\subsection{Two weighted Holm procedures}
We now describe two weighted Holm procedures that are widely used in practice. The key difference lies in how the hypotheses are ordered: WHP orders by weighted $p$-values, while WAP orders by raw $p$-values. This choice affects both their properties and performance.  

\begin{procedure}\citep{holm1979simple}. \label{whp}  
\textbf{WHP.} ~Reject $H^*_{(i)}$ if  
\begin{equation}
    \tilde{P}_{(j)} \leq \frac{\alpha}{\sum_{k=j}^{m} w^*_{(k)}}, \quad j=1,\ldots,i.
\end{equation}
Here, the hypotheses are ordered according to the weighted $p$-values. By incorporating weights into the ordering, WHP gives priority to hypotheses deemed more important, extending the idea of the weighted Bonferroni test.  
\end{procedure}  

\begin{procedure}\citep{benjamini1997multiple}. \label{wap}  
\textbf{WAP.} ~Reject $H_{(i)}$ if  
\begin{equation}
    P_{(j)} \leq \frac{w_{(j)}}{\sum_{k=j}^{m} w_{(k)}} \alpha, \quad j=1,\ldots,i.
\end{equation}
In contrast, WAP orders hypotheses by their raw $p$-values, regardless of their assigned weights. This makes the procedure appear more objective, but as we will see, it also leads to a loss of certain monotonicity properties.  
\end{procedure}  

\begin{remark}\label{divergingweights} 
It is worth noting that when the weights exhibit highly uneven magnitudes (e.g., diverging weights such that $\sum_{k=1}^{m} w_{k} = 1$, $w_{1} \to 1$, and $w_{k+1}/w_{k} \to 0$ for $k = 1, \ldots, m-1$), WHP and WAP display markedly different behavior. Under such configurations, ordering hypotheses by weighted $p$-values causes WHP to effectively reduce to a fixed-sequence testing procedure, with the testing order largely determined by the weights \citep{westfall2001optimally}. In contrast, WAP continues to order hypotheses according to their raw $p$-values, while its rejection thresholds depend on the weights. This decoupling of ordering and stopping rules highlights an important conceptual distinction between the two procedures.
\end{remark}

\smallskip
To compare WHP and WAP, we focus on four key aspects:  
(i) their CTPs, emphasizing monotonicity and consonance;  
(ii) graphical representations that offer intuitive visualization;  
(iii) adjusted $p$-values and adjusted weighted $p$-values; and  
(iv) optimality, which characterizes the decision rules and their relation to the global significance level~$\alpha$.

\section{Underlying CTPs and theoretical results}\label{ctps}

In this section, we study the underlying CTPs of the weighted Holm procedures. This perspective clarifies their statistical properties and enables direct comparison. We also examine their power, monotonicity, and consonance, and illustrate the differences through examples.  

\subsection{CTP representations of WAP and WHP}

We first show that both WAP and WHP can be expressed as CTPs. This provides a unified framework for studying their properties.  

For any intersection hypothesis $H_I=\bigcap_{i \in I} H_i$, with $I \subseteq\{1, \ldots, m\}$, let
$$
P_{(1)}^I=\min _{i \in I} P_i
$$
denote the smallest $p$-value in $I$, and let $w_{(1)}^I$ be its associated weight.

The local test for WAP rejects $H_I$ if  
\begin{equation}\label{waplocal}
    P_{(1)}^I \leq \frac{w_{(1)}^I}{\sum_{i \in I} w_i}\alpha.
\end{equation}  

The type I error rate of this test satisfies  
\[
\Pr\!\left(P_{(1)}^I \leq \tfrac{w_{(1)}^I}{\sum_{i \in I} w_i}\alpha \right) 
   \leq \sum_{i \in I} \Pr\!\left(P_i \leq \tfrac{w_i}{\sum_{i \in I} w_i}\alpha \right) \leq \alpha ,
\]  
so by the closure principle, the resulting procedure controls the FWER.  

\begin{proposition}\label{eqwap}
The CTP with local tests defined in \eqref{waplocal} is equivalent to the WAP of \citet{benjamini1997multiple}.  
\end{proposition}

\smallskip
For WHP, each intersection hypothesis $H_I$ is tested using a weighted Bonferroni test, which rejects $H_I$ if
\begin{equation}\label{whplocal}
    P_i \leq \frac{w_i}{\sum_{i \in I} w_i}\alpha \quad \text{for some } i \in I.
\end{equation}  

This test is valid since  
\[
\Pr\!\left( \cup_{i \in I} \{ P_i \leq \tfrac{w_i}{\sum_{i \in I} w_i}\alpha \}\right) 
  \leq \sum_{i \in I} \Pr\!\left(P_i \leq \tfrac{w_i}{\sum_{i \in I} w_i}\alpha \right) \leq \alpha .
\]  

\begin{proposition} \citep{westfall2001optimally}. \label{eqwhp}
The CTP with local tests defined in \eqref{whplocal}, that is, the weighted Bonferroni test, is equivalent to \citet{holm1979simple}'s  WHP.  
\end{proposition}

\subsection{Power, monotonicity, and consonance}

By comparing the local tests of WAP and WHP, we can establish a power dominance result.  

\begin{theorem}\label{thm1}
For any hypothesis $H_i$, if it is rejected by WAP, it is also rejected by WHP. Thus, WHP is uniformly more powerful than WAP under arbitrary dependence. 
\end{theorem}

\noindent \textit{Proof}. 
For each intersection hypothesis $H_I = \cap_{i \in I} H_i$, the rejection event  
\[
\{ P_{(1)}^I \leq \tfrac{w_{(1)}^I}{\sum_{i \in I} w_i}\alpha \}
\]
of WAP’s local test implies  
\[
\bigcup_{i \in I}\{ P_i \leq \tfrac{w_i}{\sum_{i \in I} w_i}\alpha \},
\]
the rejection event of WHP’s local test. Hence, by the closure principle and Propositions~3.1 and~3.2, any $H_i$ rejected by WAP must also be rejected by WHP. \hfill$\square$

\smallskip
Notably, the converse does not hold in general; that is, there exist configurations in which WHP rejects certain hypotheses that WAP fails to reject. This phenomenon is illustrated in Example~\ref{ex2} in Section~\ref{ctps}.

\begin{remark}
The critical values of WAP lack the monotone property \citep{benjamini1997multiple}, which causes power loss. If the weights do not alter the ordering of $p$-values, i.e., $\tilde{P}_{(i)} = P_{(i)}/w_{(i)}$, then WHP and WAP coincide since  
\[
\tilde{P}_{(i)} \leq \frac{\alpha}{\sum_{k=i}^m w_{(k)}^*} \iff 
P_{(i)} \leq \frac{w_{(i)}}{\sum_{k=i}^m w_{(k)}}\alpha .
\]  
\end{remark}

We now formalize two notions of monotonicity.  

\begin{definition}[Monotone property; \citeauthor{dmitrienko2009multiple}, \citeyear{dmitrienko2009multiple}]\label{m1}
A $p$-value–based MTP is said to be $p$-value monotone if lowering one or more $p$-values never decreases the number of rejections.  
\end{definition}

\begin{definition}[Monotonicity condition; \citeauthor{hommel2007powerful}, \citeyear{hommel2007powerful}]\label{m2}
For a CTP, let $\alpha_i(I)$, $i \in I \subseteq M=\{1,\ldots,m\}$, denote the intersection-specific critical values, which satisfy  
\[
\sum_{i \in I} \alpha_i(I) \leq \alpha.
\]
The procedure is said to satisfy the \emph{monotonicity condition} if  
\[
\alpha_i(I) \leq \alpha_i(J) \quad \text{for all } i \in J \subset I \subseteq M.
\]
\end{definition}

\begin{remark}
The quantities $\alpha_i(I)$, for $i \in I \subseteq M=\{1,\ldots,m\}$, in Definition~\ref{m2} denote intersection-specific critical values assigned to the component hypotheses within the intersection, where $I$ is the index set defining the intersection. These quantities arise from the closed testing procedure and depend on the particular intersection hypothesis $H_I$.

In contrast, the quantities $\alpha_i$, for $i \in I \subseteq M=\{1,\ldots,m\}$, introduced in Section~\ref{graphs}, denote local significance levels assigned to individual hypotheses in the initial graph (Step 1) and in the updated graphs (Step 4) of the graphical algorithm. In this context, $I$ represents the set of remaining (unrejected) hypotheses at the corresponding stage of the procedure.
\end{remark}

\smallskip
The WHP is $p$-value monotone, and its CTP representation satisfies the monotonicity condition, making it consonant. By contrast, WAP is neither $p$-value monotone nor monotonic in its CTP representation. Nevertheless, its closed testing representation remains consonant—a somewhat surprising property.

\begin{proposition}\label{wapconsonant}
The CTP corresponding to WAP is consonant, despite WAP itself not being monotone under Definitions~\ref{m1} or~\ref{m2}.  
\end{proposition}

\subsection{Illustrative examples}

The following examples demonstrate when WHP and WAP coincide and when WHP outperforms WAP.  

\begin{example}\label{ex1} \rm
Consider testing three hypotheses $H_1, H_2$, and $H_3$ with weights $w_1=1$, $w_2=2$, and $w_3=3$. Suppose the raw $p$-values are $p_1=0.01$, $p_2=0.03$, and $p_3=0.09$, yielding weighted $p$-values $\tilde{p}_1=0.01$, $\tilde{p}_2=0.015$, and $\tilde{p}_3=0.03$. Since the ordering of the raw and weighted $p$-values is identical, WHP and WAP follow the same testing sequence and yield identical results.
\end{example}

\begin{example}\label{ex2} \rm
Consider the same hypotheses and weights as in the previous example with $\alpha=0.05$. Let $p_1=0.01$, $p_2=0.014$, and $p_3=0.3$, yielding weighted $p$-values $\tilde{p}_1=0.01$, $\tilde{p}_2=0.007$, and $\tilde{p}_3=0.1$. In this case, the ordering of raw and weighted $p$-values differs.  

WAP compares $p_1, p_2, p_3$ against thresholds $\alpha/6=0.0083$, $2\alpha/5=0.02$, and $\alpha=0.05$, rejecting none. In contrast, WHP compares $\tilde{p}_2, \tilde{p}_1, \tilde{p}_3$ against thresholds $\alpha/6=0.0083$, $\alpha/4=0.0125$, and $\alpha/3=0.017$, and rejects $H_1$ and $H_2$. Hence, WHP yields more rejections in this setting.
\end{example}

\section{Graphical representations of WHP and WAP}\label{graphs}

\noindent In the graphical approach \citep{bretz2009graphical}, each hypothesis is represented by a node, and the global significance level $\alpha$ is initially distributed across the nodes. The relationships among hypotheses are encoded through transition coefficients, which specify how local significance levels are reallocated when a hypothesis is rejected. Formally, let the initial local levels be $\boldsymbol{\alpha} = (\alpha_1,\ldots,\alpha_m)$ with $\sum_{i=1}^m \alpha_i \leq \alpha \in (0,1)$. Define $M = \{1,\ldots,m\}$ and $G = (g_{ij})_{m \times m}$, where $g_{ij}$ is the fraction of $\alpha_i$ transferred to $H_j$ if $H_i$ is rejected, subject to $0 \leq g_{ij} \leq 1$, $g_{ii}=0$, and $\sum_{j=1}^m g_{ij} \leq 1$ for all $i$. The pair $(G, \boldsymbol{\alpha})$, together with the updating algorithm, defines a MTP.  

\smallskip
A convenient choice of the initial transition matrix is
\[
\boldsymbol{G} =
\begin{bmatrix}
0 & \tfrac{w_2}{\sum_{j=2}^m w_j} & \cdots & \tfrac{w_m}{\sum_{j=2}^m w_j} \\
\tfrac{w_1}{\sum_{j\neq 2}^m w_j} & 0 & \cdots & \tfrac{w_m}{\sum_{j\neq 2}^m w_j} \\
\vdots & \vdots & \ddots & \vdots \\
\tfrac{w_1}{\sum_{j=1}^{m-1} w_j} & \tfrac{w_2}{\sum_{j=1}^{m-1} w_j} & \cdots & 0
\end{bmatrix},
\]
with initial local levels $\alpha_i = \tfrac{w_i}{\sum_{j=1}^m w_j}\alpha$ for $i=1,\ldots,m$, which guarantees $\sum_{j=1}^m g_{ij} = 1$ for all $i$ and $\sum_{i=1}^m \alpha_i = \alpha$. Here, the initial local significance levels $\alpha_i, i=1, \ldots, m$, are specified for the graphical approaches underlying both the WHP and WAP procedures, which are described below.

\begin{tcolorbox}[title=Graphical Algorithms, colback=white, colframe=black!50]
\begin{minipage}[t]{0.48\linewidth}
\subsubsection*{WHP}
\begin{enumerate}
  \item Initialize $I = M$.
  \item Let $j = \arg\min_{i \in I} \tilde{p}_i$.
  \item If $p_j \leq \alpha_j$, reject $H_j$; otherwise stop.
  \item Update:
  \[
    I \gets I \setminus \{j\}, \quad 
    \alpha_l \gets \alpha_l + \alpha_j g_{jl} \;(l \in I),
  \]
  \[
    g_{lk} \gets \frac{g_{lk} + g_{lj} g_{jk}}{1 - g_{lj} g_{jl}} \;(l,k \in I, l \ne k).
  \]
  \item If $I \neq \varnothing$, return to Step~2; otherwise stop.
\end{enumerate}
\end{minipage}\hfill
\begin{minipage}[t]{0.48\linewidth}
\subsubsection*{WAP}
\begin{enumerate}
  \item Initialize $I = M$.
  \item Let $j = \arg\min_{i \in I} p_i$.
  \item If $p_j \leq \alpha_j$, reject $H_j$; otherwise stop.
  \item Update the graph as in Step~4 of WHP.
  \item If $I \neq \varnothing$, return to Step~2; otherwise stop.
\end{enumerate}
\end{minipage}
\end{tcolorbox}

\smallskip
\noindent The graphical representations of WHP and WAP are displayed side by side above.  
Their \textbf{essential distinction} lies in that WHP operates on weighted $p$-values $(\tilde{p}_i)$, whereas WAP is based on the raw $p$-values $(p_i)$ (see Step 2 of the graphical algorithms for both procedures). 

\begin{proposition}\label{eqgraphs}
The WHP ( \citeauthor{holm1979simple}, \citeyear{holm1979simple})  
and the WAP (\citeauthor{benjamini1997multiple}, \citeyear{benjamini1997multiple})  
are equivalent to the graphical algorithms described above,  
in the sense that they yield identical rejection decisions for any configuration of $p$-values.  
\end{proposition}

\smallskip
Graphical formulations of WHP have previously been considered in \citet{alosh2014advanced} and \citet{guilbaud2018simultaneous}, though only for a small number of hypotheses (e.g., $m=3$). Here, we extend this representation to arbitrary $m$. In addition, we introduce a graphical representation for WAP based on the raw $p$-values, which, to the best of our knowledge, has not been previously developed in the literature. Presenting these representations within a common graphical setup enables a direct comparison between WHP and WAP, thereby clarifying their structural differences and connections.

\begin{remark}\rm
The graphical approach underlying WAP lies beyond the original framework of \citet{bretz2009graphical},  
because its local tests are not weighted Bonferroni tests and thus need not satisfy the usual monotonicity conditions.  
This highlights the need for a broader graphical framework capable of encompassing WAP and related procedures.  
\end{remark}

\begin{example}\label{ex3}\rm
Consider the same setting as in Example~\ref{ex2}.  
Figures~\ref{fig:wapex3} and \ref{fig:whpex3} display the graphical representations of WAP and WHP, respectively.  
The conclusions mirror those of Example~\ref{ex2}: WAP rejects no hypotheses, whereas WHP rejects $H_1$ and $H_2$.  
\end{example}

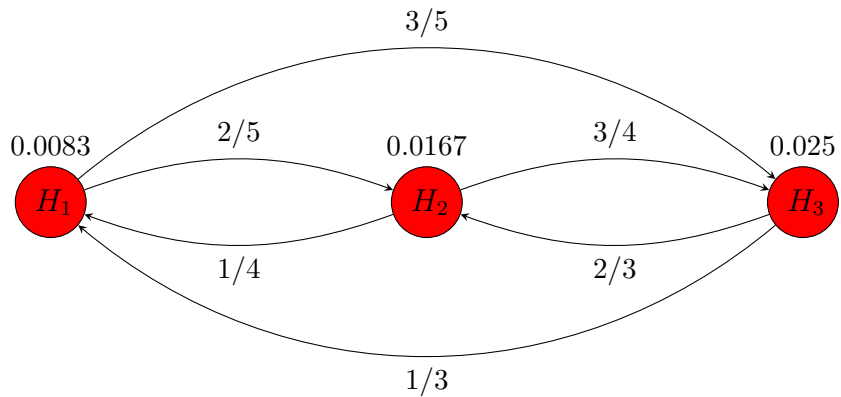
\begin{figure}[H]
\centering
\begin{tikzpicture}[node distance=5.0cm, auto, >=stealth]
\node[block, fill=red, label={\small $0.0167$}] (b) {$H_2$};
\node[block,fill=red, label={\small $0.0083$}] (a) [left of=b, node distance=5cm] {$H_1$};
\node[block, fill=red, label={\small $0.025$}] (c) [right of=b, node distance=5cm] {$H_3$};
\draw[->] (a) to [out=20,in=160] node[above] {\small $2/5$} (b);
\draw[->] (b) to [out=-160,in=-20] node[below] {\small $1/4$} (a);
\draw[->] (b) to [out=20,in=160] node[above] {\small $3/4$} (c);
\draw[->] (c) to [out=-160,in=-20] node[below] {\small $2/3$} (b);
\draw[->] (a) to [out=40,in=140] node[above] {\small $3/5$} (c);
\draw[->] (c) to [out=-140,in=-40] node[below] {\small $1/3$} (a);
\end{tikzpicture}
\caption{Graphical representation of WAP with $\alpha=0.05$, weights $w_i=i$ for $i=1,2,3$, and initial allocation $\boldsymbol{\alpha} = (\alpha/6, \alpha/3, \alpha/2)$. 
Given the raw $p$-values $p_1=0.01$, $p_2=0.014$, and $p_3=0.3$, no hypotheses are rejected under WAP. 
Rejected nodes are shown in yellow, and nodes that remain non-rejected at the final stage are shown in red.}
\label{fig:wapex3}
\end{figure}

\begin{figure}[H]
\centering
\includegraphics[width=0.95\linewidth]{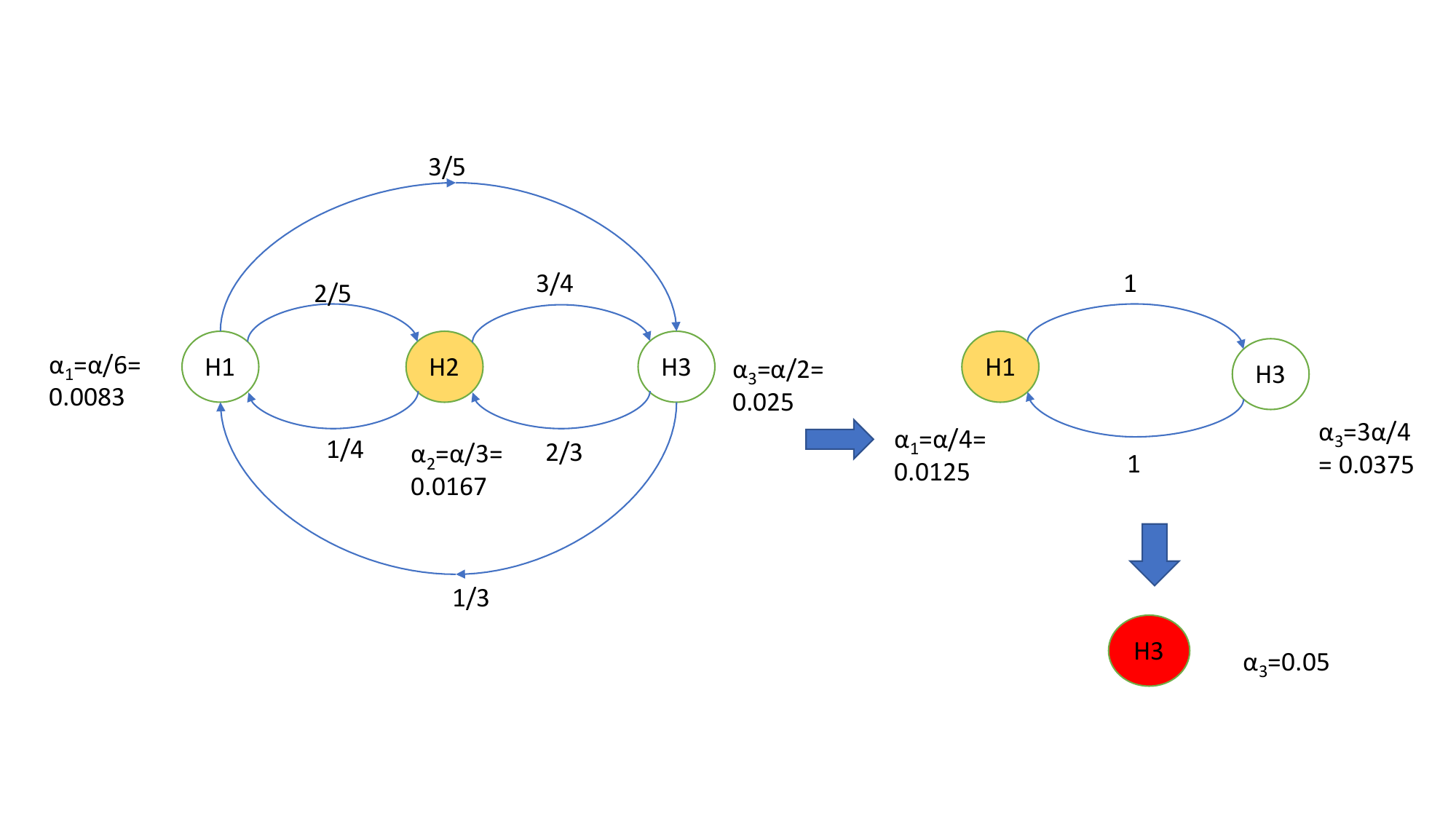} \vskip -30pt
\caption{Graphical representation of WHP under the same setting as Figure~\ref{fig:wapex3}. Given the raw $p$-values $p_1=0.01$, $p_2=0.014$, and $p_3=0.3$, the corresponding weighted $p$-values are $\tilde{p}_1=0.01$, $\tilde{p}_2=0.007$, and $\tilde{p}_3=0.1$. Under WHP, hypotheses $H_1$ and $H_2$ are rejected. Rejected nodes are shown in yellow, and nodes that remain non-rejected at the final stage are shown in red.}
\label{fig:whpex3}
\end{figure}

\section{Adjusted $p$-values}\label{adjs}

Adjusted $p$-values provide an alternative implementation of MTPs.  
Instead of comparing raw or weighted $p$-values to stepwise thresholds, one can compare adjusted $p$-values directly with the global significance level $\alpha$.  
This leads to decisions identical to those of the original procedures.  
In this section, we derive the adjusted weighted $p$-values for WHP and the adjusted $p$-values for WAP.  

\smallskip
For WHP, the adjusted weighted $p$-values $\tilde{P}_{(i)}^{\,adj}$ corresponding to the ordered hypotheses $H_{(i)}^*$ are
\[
\tilde{P}_{(i)}^{\,adj} = 
\begin{cases}
  \min\!\left\{\tilde{P}_{(1)} \sum_{k=1}^{m} w_{(k)}^*, \, 1\right\}, & i=1, \\[6pt]
  \max\!\left\{\tilde{P}_{(i)} \sum_{k=i}^{m} w_{(k)}^*, \, \tilde{P}_{(i-1)}^{\,adj}\right\}, & i=2,\ldots,m. 
\end{cases}
\]

\smallskip
For WAP, the adjusted $p$-values $P_{(i)}^{\,adj}$ corresponding to the ordered hypotheses $H_{(i)}$ are
\[
P_{(i)}^{\,adj} = 
\begin{cases}
  \min\!\left\{\dfrac{P_{(1)}}{w_{(1)}} \sum_{k=1}^{m} w_{(k)}, \, 1\right\}, & i=1, \\[6pt]
  \max\!\left\{\dfrac{P_{(i)}}{w_{(i)}} \sum_{k=i}^{m} w_{(k)}, \, P_{(i-1)}^{\,adj}\right\}, & i=2,\ldots,m. 
\end{cases}
\]

\begin{proposition}\label{adjprop}
For every hypothesis $H_i$, $i=1,\ldots,m$, the adjusted weighted $p$-value under WHP is less than or equal to the corresponding adjusted $p$-value under WAP.  
Consequently, any hypothesis rejected by WAP will also be rejected by WHP.  
\end{proposition}

Proposition~\ref{adjprop} confirms Theorem~\ref{thm1}: every rejection made by WAP is also made by WHP. Thus, WHP uniformly dominates WAP while maintaining strong FWER control, and in terms of adjusted $p$-values, it constitutes a uniformly more refined procedure.

\begin{example}\label{ex4}\rm
Consider the setting of Example~\ref{ex2}.  
From the formulas above we obtain
\[
P_1^{\,adj} = P_2^{\,adj} = 0.06, \qquad P_3^{\,adj} = 0.3,
\]
\[
\tilde{P}_1^{\,adj} = \tilde{P}_2^{\,adj} = 0.042, \qquad \tilde{P}_3^{\,adj} = 0.3.
\]
With $\alpha=0.05$, WHP rejects $H_1$ and $H_2$, while WAP rejects none.  
This outcome is consistent with Example~\ref{ex2} and further illustrates the stronger performance of WHP.  
\end{example}

\section{Optimality of the WHP and WAP}\label{optimality}

A central question in small-scale multiple testing is whether a given procedure is \emph{optimal},  
in the sense that it cannot be improved without compromising control of the FWER \citep{lehmann2005generalizations, gordon2008optimality, lehmann2011optimality}.  
To the best of our knowledge, no previous work has examined the optimality of WHP or WAP.  
In this section, we establish such results for both procedures.  
Throughout, we assume that the weights $w_i$, $i=1,\ldots,m$, are fixed in advance  
and common to all weighted $p$-value–based procedures under consideration.  

\subsection{Optimality of WHP}

We first show that WHP is optimal in a broad sense.  
By constructing suitable joint distributions for the $p$-values $P_1,\ldots,P_m$ under arbitrary dependence,  
one can verify that the FWER bound of $\alpha$ is sharp.  

\begin{theorem}\label{optwhp}
The WHP (\citeauthor{holm1979simple}, \citeyear{holm1979simple}) is optimal in the sense that no critical value can be increased without violating FWER control.  
\end{theorem}

\smallskip
A stronger form of optimality arises when we restrict attention to weighted $p$-value step-down procedures.  
Let $\alpha_1 \leq \cdots \leq \alpha_m$ be a nondecreasing sequence of critical values.  
A weighted $p$-value step-down procedure rejects $H_{(i)}^*$ if and only if
\[
\tilde{P}_{(j)} \leq \alpha_j, \quad j=1,\ldots,i,
\]
where $\tilde{P}_{(1)} \leq \cdots \leq \tilde{P}_{(m)}$ are the ordered weighted $p$-values $\tilde{P}_i = P_i / w_i$.  

\begin{definition}[Dominance]
Let $\mathcal{M}$ and $\mathcal{M}'$ be two weighted $p$-value–based procedures, with rejection sets $\mathcal{M}(\mathbf{\tilde{P}})$ and $\mathcal{M}'(\mathbf{\tilde{P}})$ for a vector of weighted $p$-values $\mathbf{\tilde{P}}$.  
We say that $\mathcal{M}'$ dominates $\mathcal{M}$, written $\mathcal{M}' \succeq \mathcal{M}$, if
\[
\mathcal{M}'(\mathbf{\tilde{P}}) \supseteq \mathcal{M}(\mathbf{\tilde{P}}) \quad \text{for all } \mathbf{\tilde{P}}.
\]
\end{definition}

\begin{proposition}\label{dominate}
Let $\mathcal{M}^w$ be any weighted $p$-value step-down procedure with $\mathrm{FWER} \leq \alpha$.  
Then $\mathcal{M}^w \preceq \mathrm{WHP}$.  
\end{proposition}

Proposition~\ref{dominate} demonstrates that, within the class of weighted $p$-value step-down procedures, WHP is uniformly most powerful.  

\subsection{Optimality of WAP}

We now turn to WAP.  
In contrast to WHP, WAP is not uniformly optimal; rather, it attains optimality only under specific restrictions on the weight ratios.  

\begin{proposition}\label{optwap}
The WAP (\citeauthor{benjamini1997multiple}, \citeyear{benjamini1997multiple}) is optimal in the sense that no critical value can be increased without violating FWER control, provided that
\[
\frac{\min_i w_i}{\max_i w_i} \,\geq\, \alpha.
\]
\end{proposition}

\begin{remark}\rm
The condition $\min_i w_i / \max_i w_i \geq \alpha$ is typically not restrictive in applications.  
For instance, when $\alpha = 0.05$, it merely requires that the largest weight is at most $20$ times the smallest weight.  
Such a balance is common in practice, since weights are usually assigned to reflect relative importance of hypotheses rather than to induce extreme disparities.  
If the condition is violated, however, the WAP may fail to retain its optimality.  
\end{remark}

\smallskip
\textit{In summary, WHP enjoys unconditional optimality, whereas WAP is only conditionally optimal, depending on the relative magnitudes of the weights.}

\section{Simulation Studies}\label{simul}

We conducted simulation studies to compare the performance of WHP and WAP in terms of average power and FWER control under dependence. The dependence structure was modeled using an equicorrelated covariance matrix $\Sigma$, where all off-diagonal entries were set to the correlation $\rho$ and diagonal entries to 1. 

In each simulation, we generated $n$ independent $m$-dimensional normal vectors with covariance matrix $\Sigma$. Each component followed 
\[
Z_i \sim N(\mu_i,1), \quad i=1,\ldots,m.
\]
For each hypothesis $H_i: \mu_i=0$ versus $H_i': \mu_i>0$, one-sided $p$-values were computed using the one-sample $t$-test.

\smallskip
The simulation parameters were specified as follows: significance level $\alpha=0.05$; number of hypotheses $m \in \{5,10,15\}$; proportion of true nulls $\pi_0 \in \{0.2,0.4,0.6,0.8\}$; correlation $\rho \in \{0,0.1,\ldots,0.9\}$; and sample size $n=15$. Among the $m$ hypotheses, $m\pi_0$ had $\mu_i=0$, while the remaining hypotheses were set to $\mu_i=0.7$. Each configuration was replicated 5{,}000 times. These settings were chosen to systematically examine the effects of dimensionality, sparsity, and correlation strength on the relative performance of the procedures.
 
\smallskip
Three FWER-controlling procedures were applied to the $m$ hypotheses: the conventional (unweighted) Holm procedure, and the two weighted Holm procedures, WHP and WAP. Four weight settings were considered, where weights corresponding to true nulls and false nulls were independently generated from two different uniform distributions:
\begin{enumerate}
\item $w_i \sim U(1,2)$ for true nulls and $w_i \sim U(6,10)$ for false nulls (highly informative weights).
\item $w_i \sim U(1,2)$ for true nulls and $w_i \sim U(2,10)$ for false nulls (moderately informative weights).
\item $w_i \sim U(1,2)$ for true nulls and $w_i \sim U(2,6)$ for false nulls (weakly informative weights).
\item $w_i \sim U(1,6)$ for both true and false nulls (no separation; non-informative weights).
\end{enumerate}

\smallskip
These four scenarios were designed to represent settings with highly informative, moderately informative, weakly informative, and non-informative weights, thereby allowing us to assess how the relative performance of WHP and WAP depends on the informativeness of the weight specification.

\begin{figure}[htbp]
    \centering
    \begin{subfigure}{0.49\linewidth}
        \centering
        \includegraphics[width=\linewidth]{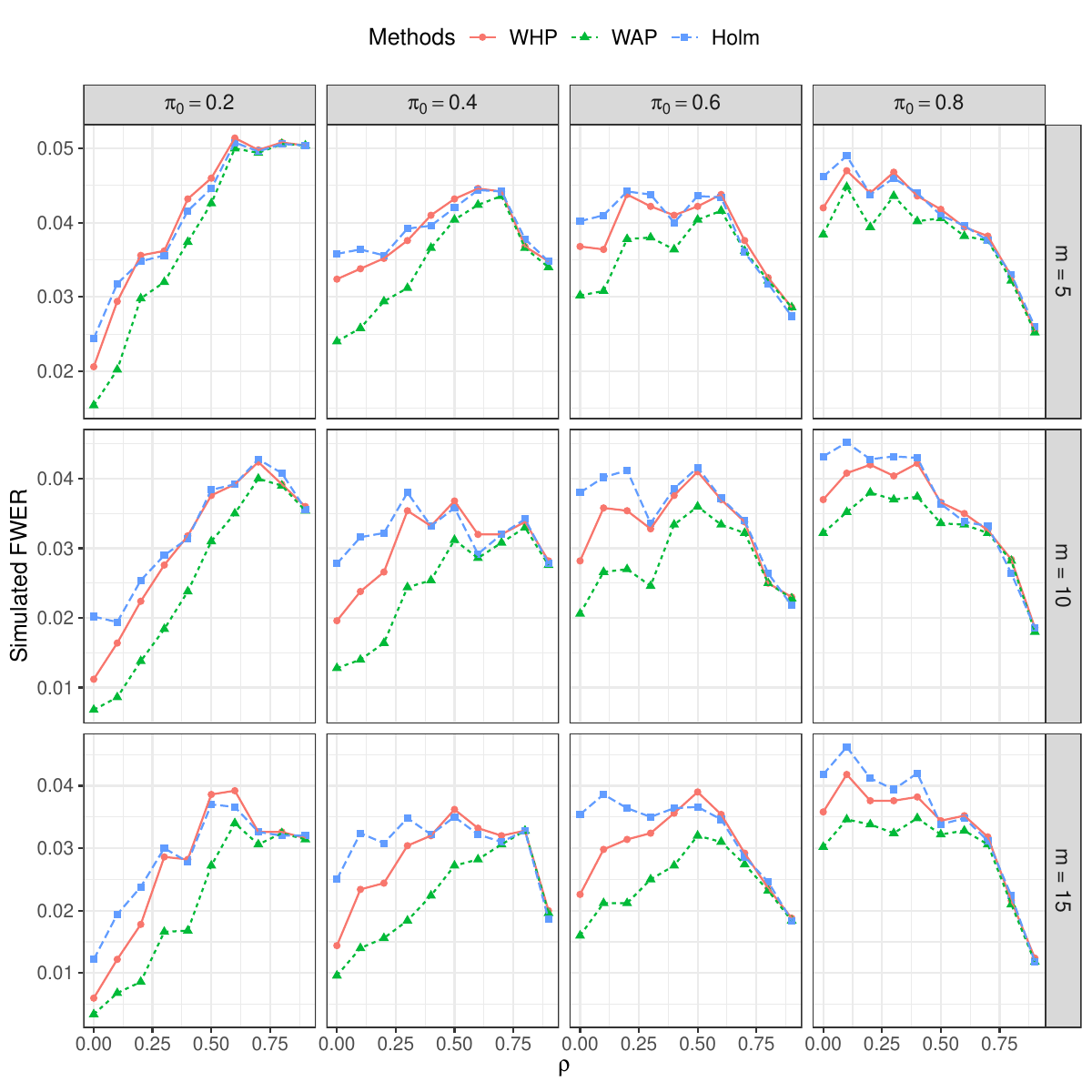}
        \caption{FWER: WHP vs.\ WAP.}
        \label{fig:fwer_12210}
    \end{subfigure}\hfill
    \begin{subfigure}{0.49\linewidth}
        \centering
        \includegraphics[width=\linewidth]{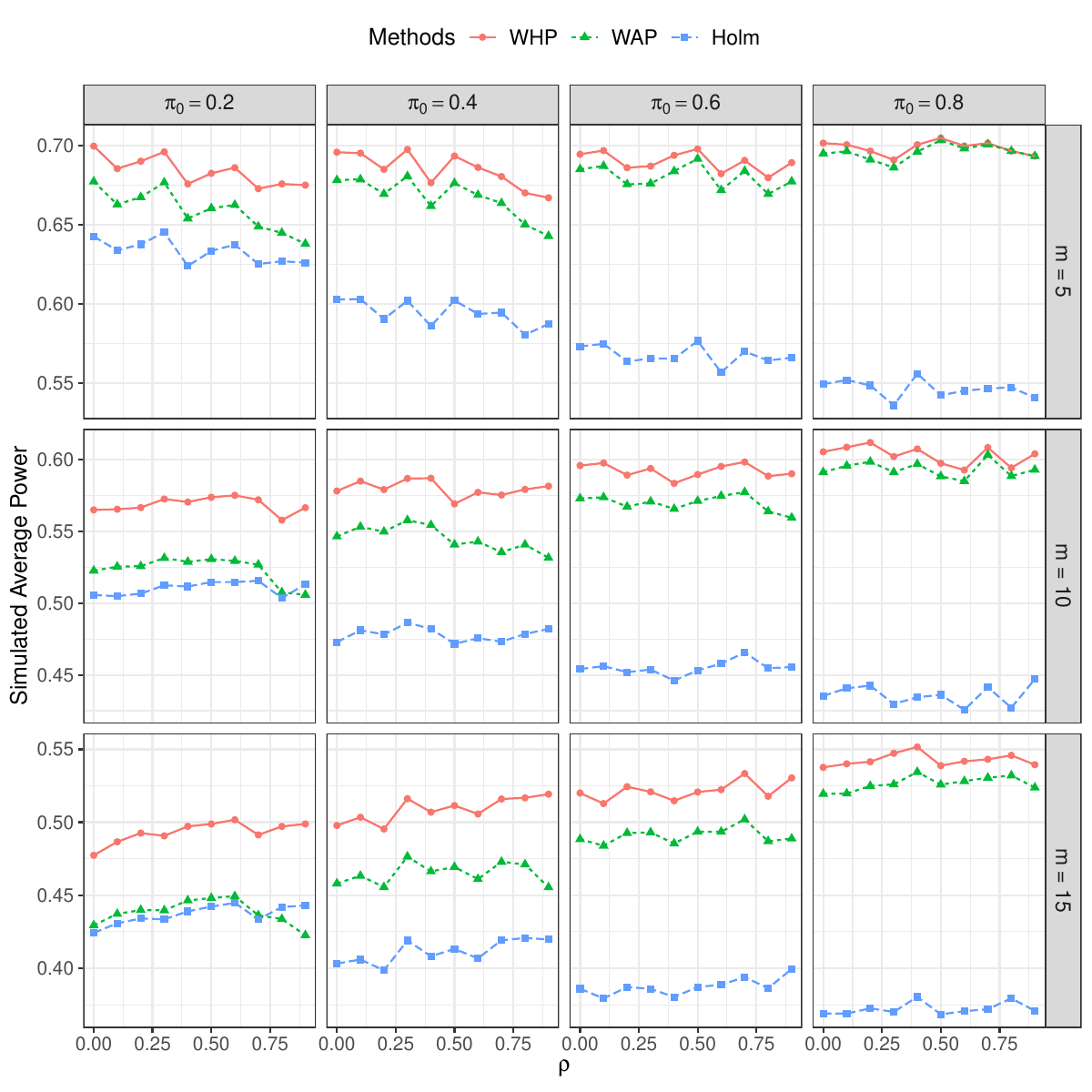}
        \caption{Average power: WHP vs.\ WAP.}
        \label{fig:pow_12210}
    \end{subfigure}
    \caption{Simulated performance under weight setting $U(1,2)$ for true nulls and $U(2,10)$ for false nulls, representing a moderately informative scenario.}
    \label{fig:sim_weights_210}
\end{figure}

\begin{figure}[htbp]
    \centering
    \begin{subfigure}{0.49\linewidth}
        \centering
        \includegraphics[width=\linewidth]{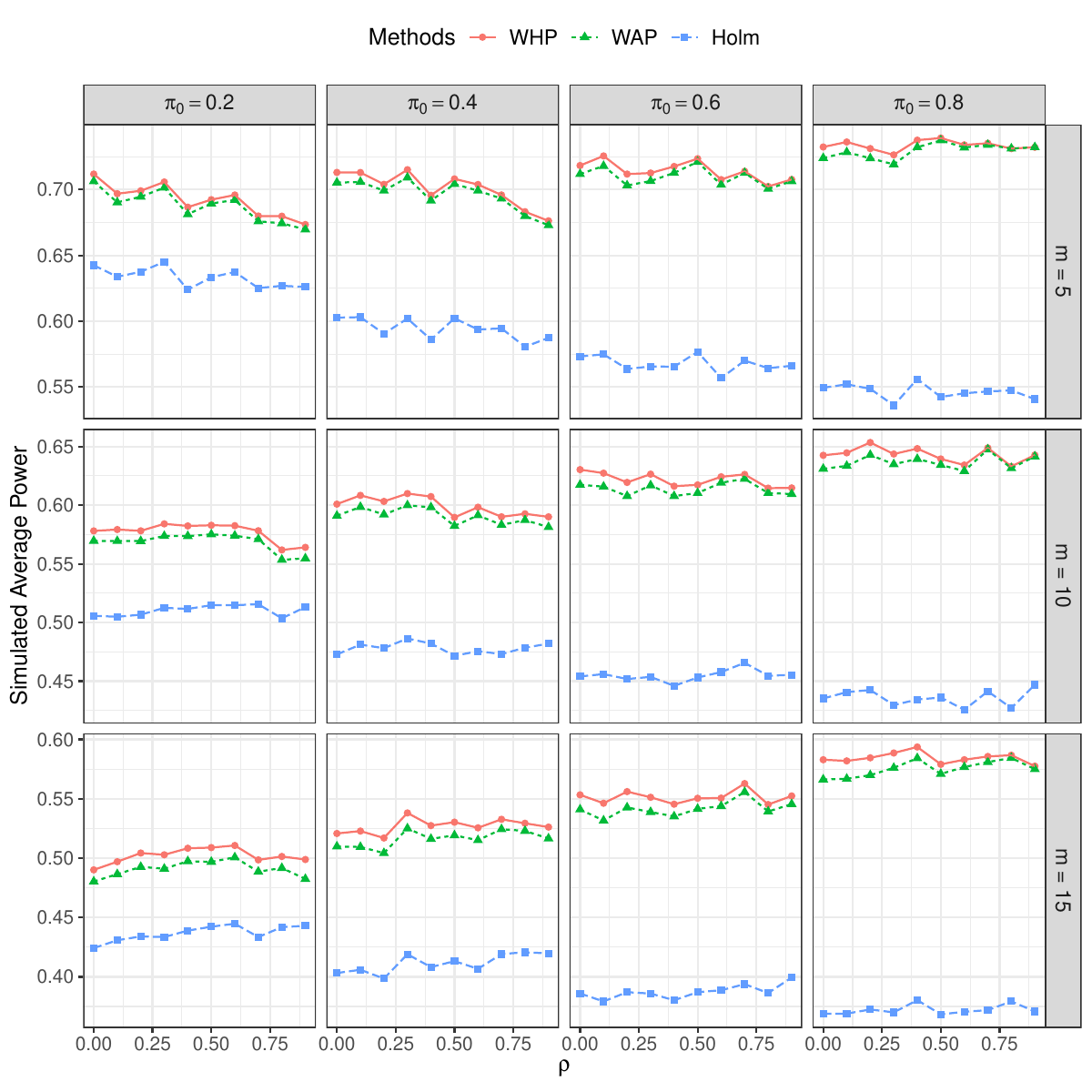}
        \caption{Power: $U(1,2)$ vs.\ $U(6,10)$.}
        \label{fig:pow_12610}
    \end{subfigure}\hfill
    \begin{subfigure}{0.49\linewidth}
        \centering
        \includegraphics[width=\linewidth]{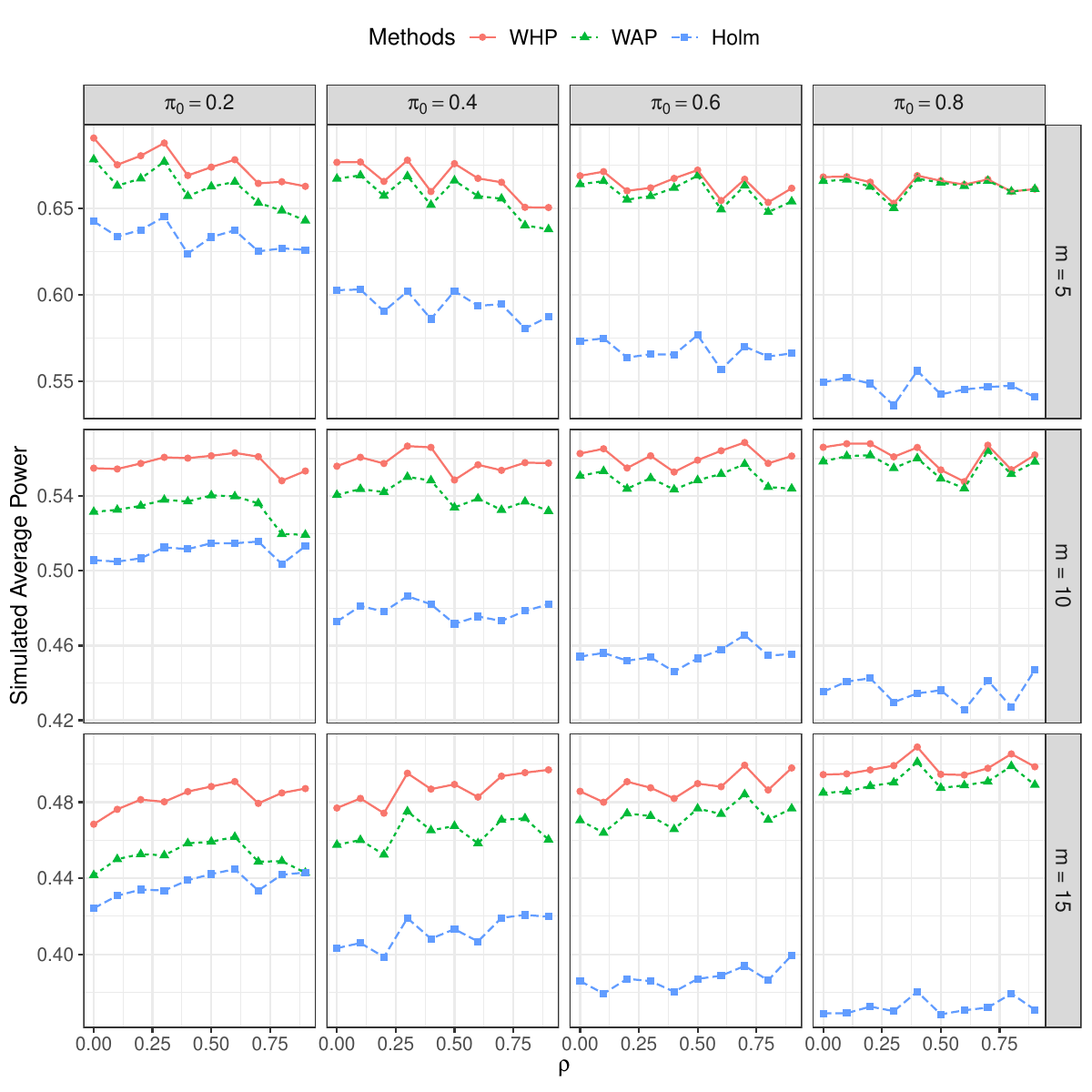}
        \caption{Power: $U(1,2)$ vs.\ $U(2,6)$.}
        \label{fig:pow_1226}
    \end{subfigure}
   \caption{Average power under two weight settings: strongly informative (high separation, left) and weakly informative (low separation, right).}
    \label{fig:sim_power_mid}
\end{figure}

\begin{figure}[htbp]
    \centering
    \begin{subfigure}{0.5\linewidth}
        \centering
        \includegraphics[width=\linewidth]{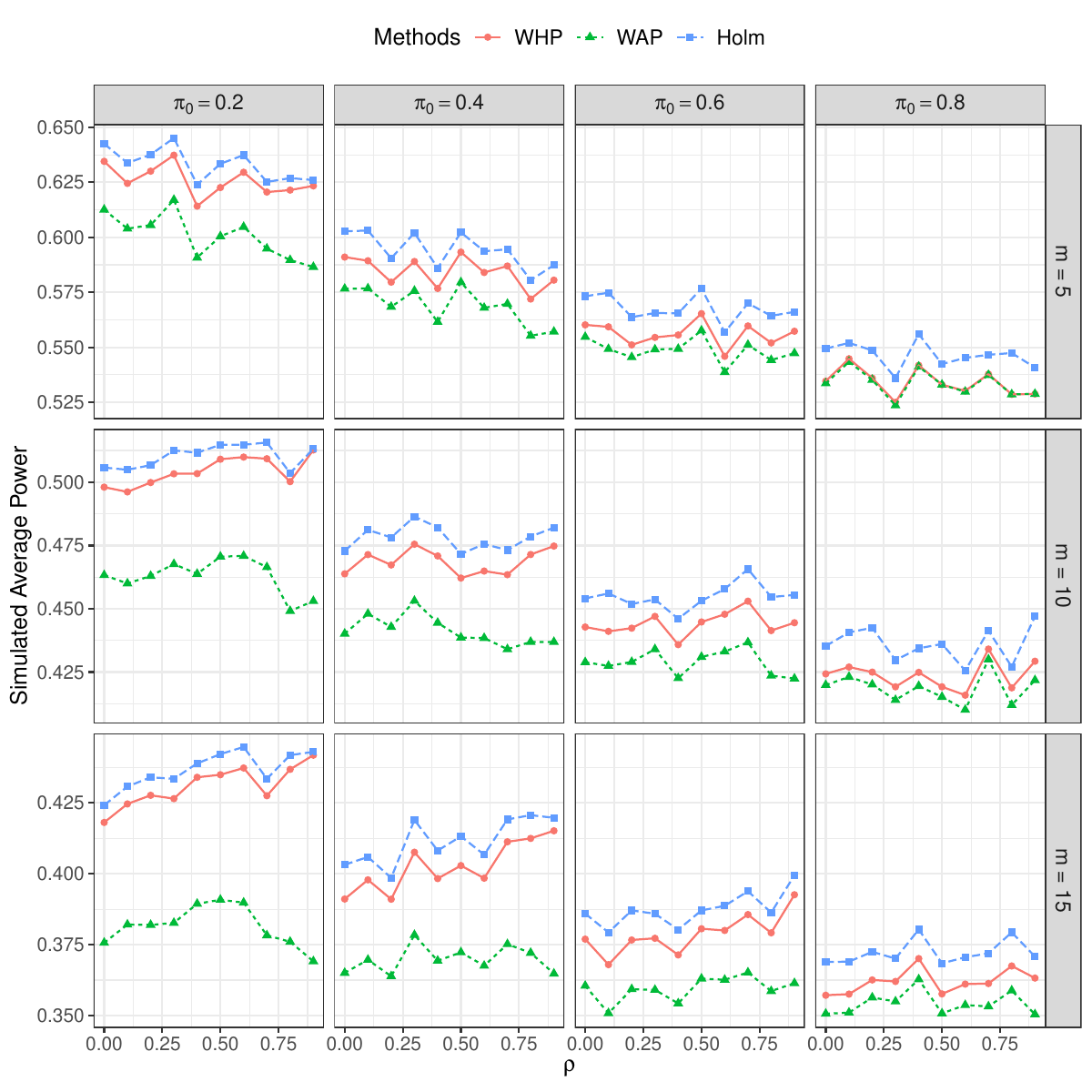}
    \end{subfigure}
    \caption{Average power of WHP, WAP, and Holm under the weight setting $U(1,6)$ for both true and false nulls, representing a non-informative (no separation) scenario.}
    \label{fig:sim_power_poor}
\end{figure}

\smallskip
We now summarize the key trends observed across different weight scenarios. Figures~\ref{fig:sim_weights_210}--\ref{fig:sim_power_poor} present the main simulation results. The findings are:  
\begin{enumerate}[(i)]
\item All three procedures (WHP, WAP, and Holm) controlled the FWER at the nominal level $\alpha$. WHP consistently achieved slightly higher (yet still valid) FWER and greater power than WAP across all scenarios. Under non-informative weights (Figure~\ref{fig:sim_power_poor}), the Holm procedure outperformed both weighted procedures, and power decreased as $\rho$ increased, particularly when the number of hypotheses was small (e.g., $m=5$). Compared with Holm, the loss of power is modest for WHP but can be substantially larger for WAP, especially when the proportion of false null hypotheses is high.
\item With informative weights (Figures~\ref{fig:sim_weights_210} and \ref{fig:sim_power_mid}), both WHP and WAP achieved higher power than Holm. Power decreased with increasing correlation, especially when the number of hypotheses was small (e.g., $m=5$). WHP was uniformly more powerful than WAP, with the advantage most pronounced under larger weight separation and higher proportions of false nulls.  
\item Proper weight specification substantially improved power for both WHP and WAP, even when the proportion of false nulls was small—contrasting with unweighted procedures such as Holm, where power typically increases with the proportion of false nulls increases.  
\end{enumerate}

\smallskip
\textit{In summary, the simulations demonstrate that WHP consistently outperforms WAP; both benefit from well-chosen weights; and poorly informative weights can negate these gains, in which case the unweighted Holm procedure may be preferable.}

\section{Real Data Analysis: Clinical Examples}\label{realdata}

To demonstrate the practical performance of WHP and WAP, we consider two clinical trial case studies.  

\subsection{ARDS Trial Example}
The first example involves patients with acute respiratory distress syndrome (ARDS) \citep{acute2000ventilation, dmitrienko2003gatekeeping}, where a novel treatment was compared with placebo. Two co-primary endpoints were evaluated:  
\begin{itemize}
    \item $H_1$: number of ventilator-free days during a 28-day period,  
    \item $H_2$: 28-day all-cause mortality.  
\end{itemize}
Two secondary endpoints were also assessed:  
\begin{itemize}
    \item $H_3$: number of ICU-free days,  
    \item $H_4$: quality of life.  
\end{itemize}

The primary endpoints were assigned unequal weights ($w_1=0.9$, $w_2=0.1$), while the secondary endpoints received equal weights ($w_3=w_4=0.5$). The raw $p$-values were $p_1=0.024$, $p_2=0.003$, $p_3=0.026$, and $p_4=0.002$, as reported in Scenario 1 of Table III in \citet{dmitrienko2003gatekeeping}.

\noindent Table~\ref{tab:table1} summarizes the raw and adjusted results under WHP and WAP.  

\begin{table}[H]
\centering
\begin{tabular}{lcccc}
\hline
Hypothesis & Weight & Raw $p$-value & Adjusted weighted $p$-value & Adjusted $p$-value \\
& & &
(WHP) & (WAP) \\
\hline
$H_1$ & 0.9 & 0.024 & 0.040 & 0.045 \\
$H_2$ & 0.1 & 0.003 & 0.040 & 0.045 \\
$H_3$ & 0.5 & 0.026 & 0.040 & 0.045 \\
$H_4$ & 0.5 & 0.002 & 0.008 & 0.008 \\
\hline
\end{tabular}
\caption{ARDS trial: weights, raw $p$-values, and adjusted results under WHP and WAP.}
\label{tab:table1}
\end{table}

\noindent Both procedures reject all four hypotheses at $\alpha = 0.05$, but WHP produces consistently smaller adjusted values, reflecting its greater power in this setting.

\subsection{Type 2 Diabetes Trial Example}
The second case study is based on a trial comparing a novel insulin formulation (Formulation A) with a standard formulation (Formulation B) in patients with Type 2 diabetes \citep{dmitrienko2007gatekeeping}. Patients were randomized into three groups (A, B, and A+B), and the primary endpoint was the mean change in hemoglobin A1c at 6 months. We examine six hypotheses with pre-specified weights $w=(6,6,5,4,2,1)$, using raw $p$-values reported in Table~4 of \citet{dmitrienko2007gatekeeping}.  

\noindent Table~\ref{tab:table2} displays the results after adjustment under WHP and WAP.  

\begin{table}[H]
\centering
\begin{tabular}{lcccc}
\hline
Hypothesis & Weight & Raw $p$-value & Adjusted weighted $p$-value & Adjusted $p$-value \\
& & &
(WHP) & (WAP) \\
\hline
$H_1$ & 6 & 0.011 & 0.0348 & 0.0348 \\
$H_2$ & 6 & 0.023 & 0.0498 & 0.0585 \\
$H_3$ & 5 & 0.006 & 0.0288 & 0.0288 \\
$H_4$ & 4 & 0.018 & 0.0498 & 0.0585 \\
$H_5$ & 2 & 0.042 & 0.0630 & 0.0630 \\
$H_6$ & 1 & 0.088 & 0.0880 & 0.0880 \\
\hline
\end{tabular}
\caption{Type 2 diabetes trial: pre-specified weights, raw $p$-values, and adjusted results under WHP and WAP.}
\label{tab:table2}
\end{table}

At $\alpha=0.05$, WHP rejects $H_1, H_2, H_3$, and $H_4$, whereas WAP rejects only $H_1$ and $H_3$. This difference arises because WHP leverages the weight distribution more efficiently, enabling detection of additional treatment effects while still maintaining strict FWER control.  

\smallskip
\textit{Together, these two case studies highlight the practical advantage of WHP over WAP: while both procedures ensure valid FWER control, WHP consistently yields smaller adjusted $p$-values and higher power, leading to more informative conclusions in clinical research.}

\section{Conclusion}\label{summary2}

The wMTPs, such as the weighted Bonferroni, weighted Holm, and weighted Hochberg methods, are widely used for controlling the FWER. However, it has not always been clear which procedure is most preferable and straightforward to apply in practice, particularly in clinical trials. In this paper, we examined two variants of the weighted Holm procedure: WHP, which orders weighted $p$-values, and WAP, which orders raw $p$-values. We compared these procedures from four perspectives: (i) representation as CTPs, (ii) graphical visualization, (iii) adjusted $p$-values, and (iv) optimality.  

\smallskip
Our theoretical results show that WHP uniformly dominates WAP in terms of power under arbitrary dependence. Furthermore, we established optimality properties: WHP is optimal among all step-down procedures based on weighted $p$-values that control the FWER, whereas WAP is optimal only under restrictive conditions on weight ratios. Simulation studies corroborated these findings, demonstrating that WHP achieves superior power while maintaining strong FWER control.  

\smallskip
Beyond power and optimality, we investigated structural properties. WHP is monotone in $p$-values and also satisfies monotonicity in its closed testing representation, which ensures consonance and enables efficient shortcut implementations. In contrast, WAP fails monotonicity under both definitions. Nevertheless, its closed testing representation remains consonant, permitting shortcut formulations. This contrast is notable: WAP does not belong to the class of CTPs with weighted Bonferroni-type local tests and lacks monotonicity, yet it still enjoys consonance.  

\smallskip
\textit{Overall, our results suggest that WHP should be the method of choice in practice. It combines stronger theoretical guarantees, more reliable empirical performance, and favorable structural properties. While future advances in generalized weighted graphical frameworks may broaden these insights, WHP currently offers the most transparent and powerful option for real-world applications.}  

A promising direction for future research is to extend weighted Holm procedures to more complex settings in which $p$-values arise from structured hypotheses, such as one- or two-way groupings. Such extensions would further enhance the applicability of weighted procedures and open new opportunities for integrating structure into multiple testing while preserving rigorous error control.

\section{Proofs} \label{append}
\subsection{Proof of Proposition \ref{eqwap}} 

\noindent Let $P_{(i)} = \min\{P_l : l \in I\}$ denote the smallest $p$-value in an index set $I$, where $i$ is the rank of this $p$-value among all $m$ $p$-values. Define  
\[
I_{(i)}^+ = \{(i), (i+1), \ldots, (m)\},
\]
the index set corresponding to the $m-i+1$ largest $p$-values. Clearly, for any set $I$ in which $P_{(i)}$ is the smallest element, we must have $I \subseteq I_{(i)}^+$.  

\smallskip
\noindent Let $\phi_I$ denote the test function of the local test for the intersection hypothesis $H_I$, i.e., $\phi_I = 1$ if $H_I$ is rejected and $\phi_I = 0$ otherwise.  

\smallskip
\noindent \textbf{Step 1.} We first show that $\phi_{I_{(i)}^+}=1$ implies $\phi_I=1$.  

\noindent If $\phi_{I_{(i)}^+}=1$, then by the local test of the CTP corresponding to WAP,  
\[
P_{(i)} \leq \frac{w_{(i)}}{\sum_{k=i}^{m} w_{(k)}} \,\alpha.
\]
Since $P_{(i)}$ is the smallest $p$-value in $I$, it follows that  
\[
\sum_{k \in I} w_{(k)} \;\leq\; \sum_{k=i}^{m} w_{(k)}.
\]
Thus,
\[
P_{(i)} = \min\{P_l : l \in I\} \;\leq\; \frac{w_{(i)}}{\sum_{k \in I} w_{(k)}} \,\alpha,
\]
which implies $\phi_I=1$.  

\smallskip
\noindent \textbf{Step 2.} By Step 1 and the closure principle, $H_{(i)}$ is rejected by the CTP if and only if  
\[
P_{(j)} \leq \frac{w_{(j)}}{\sum_{k=j}^{m} w_{(k)}} \,\alpha, \quad j=1,\ldots,i,
\]
which is precisely the rejection rule of WAP.  

\smallskip
\noindent Hence, WAP and its corresponding CTP are equivalent. \hfill$\square$  

\vspace{2mm}
\noindent The same reasoning also establishes the equivalence between WHP and its corresponding CTP, proving Proposition \ref{eqwhp} through a different approach than \citet{westfall2001optimally}.

\subsection{Proof of Proposition \ref{wapconsonant}}

\noindent Let $M=\{1,\ldots,m\}$. To establish that the CTP is consonant, it suffices to show that whenever an intersection hypothesis $H_I$ with $I \subseteq M$ is rejected, at least one of its component hypotheses $H_j$ ($j \in I$) is also rejected.

\smallskip
\noindent (i) \textbf{Case $I=M$.}  
If the global intersection hypothesis $H_M=\bigcap_{i \in M} H_i$ is rejected, then its local test satisfies  
\[
P_{(1)} \leq \frac{w_{(1)}}{\sum_{i \in M} w_i}\,\alpha.
\]
By the definition of the WAP local tests, every intersection hypothesis $H_J$ with $(1) \in J \subset M$ is also rejected. Hence, by the closure principle, the elementary hypothesis $H_{(1)}$ is rejected.  

\smallskip
\noindent (ii) \textbf{Case $I \subset M$.}  
Suppose $H_I$ is rejected by the CTP. Then we must have
\[
P_{(1)}^I \leq \frac{w_{(1)}^I}{\sum_{i \in I} w_i}\,\alpha,
\qquad
P_{(1)}^S \leq \frac{w_{(1)}^S}{\sum_{i \in S} w_i}\,\alpha
\quad \forall S \,:\, I \subset S \subseteq M,
\]
where $P_{(1)}^S$ denotes the smallest $p$-value in $S$, with corresponding weight $w_{(1)}^S$.  

\noindent Let $i_0 \in I$ be the index attaining $P_{(1)}^I$. For any subset $J$ with $i_0 \in J \subset I$, we have
\[
P_{(1)}^J = P_{(1)}^I
\quad \text{and} \quad
w_{(1)}^J = w_{(1)}^I,
\]
so that
\[
P_{(1)}^J \leq \frac{w_{(1)}^J}{\sum_{i \in J} w_i}\,\alpha.
\]
Thus $H_J$ is rejected by its local test.  

\noindent Now consider any $J \nsubseteq I$ with $i_0 \in J$, and set $S = I \cup J$. Since $I \subset S$, $H_S$ is rejected by its local test:
\[
P_{(1)}^S \leq \frac{w_{(1)}^S}{\sum_{i \in S} w_i}\,\alpha.
\]
But $P_{(1)}^S = \min\{P_{(1)}^I, P_{(1)}^J\} = P_{(1)}^J$, because $i_0 \in J$. Hence
\[
P_{(1)}^J \leq \frac{w_{(1)}^S}{\sum_{i \in S} w_i}\,\alpha 
\leq \frac{w_{(1)}^J}{\sum_{i \in J} w_i}\,\alpha,
\]
which implies that $H_J$ is rejected by its local test.  


\noindent By the closure principle, the individual hypothesis $H_{i_0}$ is therefore rejected.  

\smallskip
\noindent In conclusion, rejection of any intersection hypothesis necessarily propagates to at least one of its elementary components. Hence the WAP satisfies consonance. \hfill$\square$

\subsection{Proof of Proposition \ref{eqgraphs}}

\noindent We show that the graphical representations are equivalent to their respective procedures, WHP and WAP. Since the only difference lies in the ordering criterion (weighted $p$-values for WHP versus raw $p$-values for WAP), it suffices to establish equivalence for WHP.

\smallskip
\noindent Let $I_0=\{1,\ldots,m\}$ be the full index set. Denote by $j_1,\ldots,j_m$ the sequence of hypothesis indices selected by the graphical algorithm, and let $I_i = I_0 \setminus \{j_1,\ldots,j_i\}$ denote the set of hypotheses remaining after $i$ rejections.  

\smallskip
\noindent The proof has three parts: \textit{initialization} shows agreement with the weighted Bonferroni test; \textit{induction} verifies that the updating rules preserve this structure after each rejection; and the \textit{decision rule} coincides with WHP. The WAP case then follows analogously under raw $p$-value ordering.

---

\noindent \textbf{Step 1. Initialization.}  
At step $i=0$, the algorithm selects $P_{j_1}$ as the $p$-value corresponding to the smallest weighted $p$-value. If  
\[
P_{j_1} \leq \alpha_{j_1} = \frac{w_{j_1}}{\sum_{i \in I_0} w_i}\,\alpha,
\]
then $H_{j_1}$ is rejected; otherwise the procedure stops. This is precisely the weighted Bonferroni test applied to $I_0$.  

---

\noindent \textbf{Step 2. Induction step.}  
Suppose that after $i$ steps, the local levels are
\[
\alpha_l = \frac{w_l}{\sum_{k \in I_{i-1}} w_k}\,\alpha, \quad l \in I_{i-1},
\]
and the transition coefficients are
\[
g_{lk} = \frac{w_k}{\sum_{r \in I_{i-1}\setminus\{l\}} w_r}, \quad l,k \in I_{i-1}, \, l \neq k.
\]

\noindent When $H_{j_i}$ is rejected, the update for $\alpha_l$ gives
\[
\alpha_l \leftarrow \alpha_l + \alpha_{j_i} g_{j_i l} 
= \frac{w_l}{\sum_{r \in I_{i-1}} w_r}\,\alpha
   + \frac{w_{j_i}}{\sum_{r \in I_{i-1}} w_r}\,\alpha \cdot 
     \frac{w_l}{\sum_{r \in I_{i-1}\setminus\{j_i\}} w_r}.
\]
Simplifying,
\[
\alpha_l = \frac{w_l}{\sum_{r \in I_{i-1}\setminus\{j_i\}} w_r}\,\alpha
= \frac{w_l}{\sum_{r \in I_i} w_r}\,\alpha.
\]

\noindent For the transition coefficients, observe that
\[
\begin{aligned}
g_{lk}   &= \frac{w_k}{\sum_{r\in I_{i-1}\setminus\{l\}} w_r}, \quad
g_{l j_i} = \frac{w_{j_i}}{\sum_{r\in I_{i-1}\setminus\{l\}} w_r}, \\[6pt]
g_{j_i k} &= \frac{w_k}{\sum_{r\in I_{i-1}\setminus\{j_i\}} w_r}, \quad
g_{j_i l} = \frac{w_l}{\sum_{r\in I_{i-1}\setminus\{j_i\}} w_r}.
\end{aligned}
\]
Substituting into the update rule
\[
g_{lk} \leftarrow \frac{g_{lk}+g_{l j_i} g_{j_i k}}{1-g_{l j_i} g_{j_i l}},
\]
and simplifying, we obtain
\[
g_{lk}=\frac{w_k}{\sum_{r \in I_i\setminus\{l\}} w_r}.
\]

\noindent Thus, after each rejection, both the local levels and the transition coefficients again correspond to a weighted Bonferroni procedure on $I_i$.

---

\noindent \textbf{Step 3. Decision rule.}  
Since hypotheses are ordered by increasing $\tilde P_i = P_i/w_i$, the algorithm rejects $H_{j_i}$ if
\[
\frac{P_{j_i}}{w_{j_i}} \leq \frac{\alpha}{\sum_{k \in I_{i-1}} w_k}, \quad i=1,\ldots,\ell,
\]
which is equivalent to
\[
\tilde{P}_{(i)} \leq \frac{\alpha}{\sum_{k=i}^m w_{(k)}^*}, \quad i=1,\ldots,\ell.
\]
This is exactly the decision rule of WHP. Hence, WHP and its graphical representation are equivalent.  

---

\noindent \textbf{Step 4. WAP case.}  
For WAP, the same reasoning applies, except the ordering is by raw $p$-values ($P_{j_i} = P_{(i)}$). The decision rule becomes
\[
P_{(i)} \leq \frac{w_{(i)}}{\sum_{k=i}^m w_{(k)}}\,\alpha, \quad i=1,\ldots,\ell,
\]
which coincides with the definition of WAP.  

\smallskip
\noindent Therefore, both WHP and WAP are equivalent to their respective graphical representations. This establishes not only the formal equivalence but also provides a clear graphical interpretation of both procedures, which is often useful for intuition and implementation. \hfill$\square$

\subsection{Proof of Proposition \ref{adjprop}}  

\noindent Intuitively, the adjusted weighted $p$-values under WHP are always less than or equal to the adjusted $p$-values under WAP. Weighting reshapes the ordering of hypotheses in a way that systematically favors rejection. Equivalently, whenever WAP rejects a hypothesis, WHP—being uniformly more powerful—must also reject it. We now formalize this argument.  

\medskip
\noindent \textbf{Proof.}  
Let $P_{(1)} \leq \cdots \leq P_{(m)}$ denote the ordered raw $p$-values, and let $\tilde{P}_{(1)} \leq \cdots \leq \tilde{P}_{(m)}$ denote the ordered weighted $p$-values, where  
\[
\tilde{P}_{(r)} = \frac{P_{(i)}}{w_{(i)}}, \qquad i \in \{1,\ldots,m\}.
\]
Here, $i$ indexes the raw ordering and $r$ the weighted ordering.  

We must show that
\[
\tilde{P}_{(r)}^{\,\text{adj}} \leq P_{(i)}^{\,\text{adj}}, \qquad i=1,\ldots,m.
\]

For each $r$, define
\[
s_r = \min\Big\{\,l:\, \tilde{P}_{(r)} \leq \tfrac{P_{(l)}}{w_{(l)}} \leq \tilde{P}_{(m)}\Big\}.
\]
By construction, $s_r \leq i$. Moreover, for any $k \geq r$ we have $\tilde{P}_{(k)} \geq \tilde{P}_{(r)}$, which by minimality of $s_r$ implies $\pi(k) \geq s_r$, where $\pi$ is the permutation mapping weighted ranks to raw ranks. Thus,  
\[
\{\pi(k): k=r,\ldots,m\} \subseteq \{s_r,\ldots,m\}.
\]
Since $w_{(k)}^* = w_{(\pi(k))}$, it follows that
\[
\sum_{k=r}^m w_{(k)}^* = \sum_{k=r}^m w_{(\pi(k))} \;\leq\; \sum_{\ell=s_r}^m w_{(\ell)},
\]
because the sum of a subset of nonnegative weights cannot exceed the sum over the superset.  

Combining this with the inequality
\[
\tilde{P}_{(r)} \leq \frac{P_{(s_r)}}{w_{(s_r)}},
\]
we obtain
\[
\tilde{P}_{(r)}^{\,\text{adj}}
   = \min\!\left\{\,\tilde{P}_{(r)} \sum_{k=r}^m w_{(k)}^*, \, 1 \right\}
   \leq \min\!\left\{\,\frac{P_{(s_r)}}{w_{(s_r)}} \sum_{\ell=s_r}^m w_{(\ell)}, \, 1 \right\}
   = P_{(s_r)}^{\,\text{adj}}.
\]

Finally, by monotonicity of adjusted $p$-values,
\[
P_{(s_r)}^{\,\text{adj}} \leq P_{(i)}^{\,\text{adj}},
\]
and hence
\[
\tilde{P}_{(r)}^{\,\text{adj}} \leq P_{(i)}^{\,\text{adj}}.
\]
This completes the proof. \hfill$\square$

\subsection{Proof of Theorem \ref{optwhp}}  

\noindent To establish the optimality of WHP, it suffices to show that the bound $\text{FWER} \leq \alpha$ is sharp, i.e., equality can be attained by constructing a suitable joint distribution of the $p$-values $(P_1,\ldots,P_m)$.  

\vspace{2mm}  
\noindent We work under the least favorable configuration (LFC), where all false null $p$-values are identically zero and all true null $p$-values are i.i.d.\ $\text{Unif}(0,1)$. This setting is standard, since under arbitrary dependence the LFC yields the maximal error rate (see \citeauthor{finner2001false}, \citeyear{finner2001false}). Let $m_0$ denote the number of true nulls, and assume without loss of generality that $H_1,\ldots,H_{m_0}$ are true.  

\vspace{2mm}  
\noindent Define the following joint distribution. Select exactly one true null $H_i$, $i \in \{1,\ldots,m_0\}$, with probability  
\[
\frac{w_i}{\sum_{k=1}^{m_0} w_k}.
\]  
For this chosen $i$, set $P_i = w_i U_1$ with $U_1 \sim \text{Unif}\!\left(0,1/\sum_{k=1}^{m_0} w_k\right)$. Then  
\[
P_i \sim \text{Unif}\!\left(0, \tfrac{w_i}{\sum_{k=1}^{m_0} w_k}\right), 
\qquad \tilde{P}_i = U_1.
\]  

\noindent For the remaining true nulls $H_j$, $j \neq i$, let  
\[
P_j = w_j U_2^j, 
\qquad U_2^j \sim \text{Unif}\!\left(1/\sum_{k=1}^{m_0} w_k,\,1/w_j\right),
\]  
independent of $U_1$. Then  
\[
P_j \sim \text{Unif}\!\left(\tfrac{w_j}{\sum_{k=1}^{m_0} w_k}, 1\right),
\qquad \tilde{P}_j = U_2^j.
\]  

\noindent Unconditionally, each $P_i$ has the mixed distribution  
\[
P_i \sim \frac{w_i}{\sum_{k=1}^{m_0} w_k}\, \text{Unif}\!\left(0, \tfrac{w_i}{\sum_{k=1}^{m_0} w_k}\right) 
+ \left(1-\tfrac{w_i}{\sum_{k=1}^{m_0} w_k}\right) \text{Unif}\!\left(\tfrac{w_i}{\sum_{k=1}^{m_0} w_k}, 1\right).
\]  

\noindent \textbf{Verification of uniformity.}  
- If $u \leq w_i/\sum_{k=1}^{m_0} w_k$,  
\[
\Pr(P_i \leq u) 
= \frac{w_i}{\sum_{k=1}^{m_0} w_k}\cdot \frac{u}{w_i/\sum_{k=1}^{m_0} w_k} 
= u.
\]  
- If $u > w_i/\sum_{k=1}^{m_0} w_k$,  
\[
\Pr(P_i \leq u) 
= \frac{w_i}{\sum_{k=1}^{m_0} w_k} 
+ \Big(1-\tfrac{w_i}{\sum_{k=1}^{m_0} w_k}\Big)\frac{u-w_i/\sum_{k=1}^{m_0} w_k}{1-w_i/\sum_{k=1}^{m_0} w_k} 
= u.
\]  

\noindent Thus, each true null $p$-value is marginally $\text{Unif}(0,1)$, while exactly one weighted $p$-value satisfies $\tilde{P}_i \leq 1/\sum_{k=1}^{m_0} w_k$.  

\vspace{2mm}  
\noindent Therefore,  
\[
\begin{split}
\text{FWER}
&= \Pr(V \geq 1) \\
&= \Pr\!\left(\min_{i \in I_0} \tilde{P}_i \leq \tfrac{\alpha}{\sum_{k=1}^{m_0} w_k}\right) \\
&= \Pr\!\left(U_1 \leq \tfrac{\alpha}{\sum_{k=1}^{m_0} w_k}\right) \\
&= \frac{\alpha/\sum_{k=1}^{m_0} w_k}{1/\sum_{k=1}^{m_0} w_k} \\
&= \alpha,
\end{split}
\]  
where the second equality holds by construction, using the definitions of both the LFC and the WHP.

\vspace{2mm}  
\noindent Hence, WHP achieves the bound $\text{FWER} = \alpha$ under arbitrary dependence and is therefore unimprovable without sacrificing FWER control. \hfill$\square$

\subsection{Proof of Proposition \ref{dominate}}  

\noindent Let $\mathbf{q} = (q_1,\ldots,q_m)$ be the $p$-values for hypotheses $H_1,\ldots,H_m$, with positive weights $w_1,\ldots,w_m$. Define the weighted $p$-values  
\[
\tilde{q}_i = \frac{q_i}{w_i}, \quad i=1,\ldots,m,
\]
and let $\tilde{q}_{(1)} \leq \cdots \leq \tilde{q}_{(m)}$ denote their order, with $w_{(i)}^*$ and $H_{(i)}^*$ denoting the associated weights and hypotheses.

\vspace{2mm}
\noindent Let $\mathcal{M}^w$ be a weighted step-down procedure with nondecreasing critical values $\alpha_1 \leq \cdots \leq \alpha_m$ that controls the FWER at level $\alpha$. Suppose it rejects $H_{(1)}^*,\ldots,H_{(r)}^*$. Define  
\[
l = \sum_{k=r}^m w_{(k)}^*, 
\qquad 
\tau = \min\{\alpha_r,1/l\}, 
\qquad 
\beta = l\tau.
\]

\vspace{2mm}
\noindent Consider the least favorable configuration where all false null $p$-values are zero and all true null $p$-values are i.i.d.\ $\mathrm{Unif}(0,1)$. Suppose $m_1$ hypotheses are false; without loss of generality, let these be the first $m_1$ hypotheses. Then, for $r = m_1+1$, we obtain $l = \sum_{k=r}^m w_k$.

\vspace{2mm}
\noindent Generate the weighted $p$-values $\tilde{q}_r,\ldots,\tilde{q}_m$ for the true null hypotheses $H_r,\ldots,H_m$ as follows:
\begin{itemize}
    \item Choose $j \in \{r,\ldots,m,0\}$ with probabilities $w_r\tau,\ldots,w_m\tau,\,1-\beta$.  
    \item If $j \neq 0$, let $\tilde{q}_j \sim U[0,\tau]$ and $\tilde{q}_i \sim U[\tau,\,1/w_i]$ for $i \neq j$.  
    \item If $j=0$, let all $\tilde{q}_i \sim U[\tau,\,1/w_i]$ for $i=r,\ldots,m$.  
\end{itemize}
In either case, the induced raw $p$-values $q_i = w_i\tilde{q}_i$, $i=r,\ldots,m$, are uniform on $[0,1]$. Moreover, if $j \neq 0$, then
\begin{equation}\label{rth}
\tilde{q}_{(r)} = \min_{r \leq i \leq m} \tilde{q}_i \leq \tau \leq \alpha_r.
\end{equation}

\vspace{2mm}
\noindent For indices $1,\ldots,r-1$, we have $\tilde{q}_{(j)}=0$. Thus, with probability at least $\beta$, $\mathcal{M}^w$ rejects at least one true null. Consequently,  
\[
\mathrm{FWER} \;\geq\; \beta.
\]
Since $\mathcal{M}^w$ controls FWER at $\alpha$, it follows that $\beta \leq \alpha$, i.e.,
\[
l\tau \leq \alpha.
\]
Since $\tau = \min\{\alpha_r, 1/l\}$ and $\alpha \in (0,1)$, it follows that
\[
\tau \leq \frac{\alpha}{l}
\quad \Rightarrow \quad
\alpha_r \leq \frac{\alpha}{\sum_{k=r}^m w_k}
= \frac{\alpha}{\sum_{k=r}^m w_{(k)}^*}.
\]

\vspace{2mm}
\noindent This upper bound coincides with the rejection threshold of the WHP. Since $m_1$ is arbitrary, the inequality holds for all $r=1,\ldots,m$. Therefore,  
\[
\mathcal{M}^w \preceq \mathrm{WHP}. 
\]
\hfill$\square$



\subsection{Proof of Proposition \ref{optwap}}  

\noindent The argument parallels the proof of Theorem \ref{optwhp}, relying on the same least favorable configuration (LFC) and joint distribution of $p$-values.  

\vspace{2mm}  
\noindent Let  
\[
j^* := \arg\min_{1 \leq i \leq m_0} w_i, 
\qquad w_{j^*} = \min_{1 \leq i \leq m_0} w_i
\]  
denote the index and value of the smallest weight among the true null hypotheses. Under the LFC, exactly one true null hypothesis $H_i$, $i \in \{1,\ldots,m_0\}$, is selected with probability $\tfrac{w_i}{\sum_{k=1}^{m_0} w_k}$, and its $p$-value follows  
\[
P_i \sim \text{Unif}\!\left(0, \tfrac{w_i}{\sum_{k=1}^{m_0} w_k}\right).
\]  
For each non-selected true null, the $p$-value is distributed as  
\[
P_j \sim \text{Unif}\!\left(\tfrac{w_j}{\sum_{k=1}^{m_0} w_k}, 1\right), \quad j \in I_0 \setminus \{i\}.
\]  
Denote the collection of true null $p$-values by $\mathbb{P}_0 = \{P_1,\ldots,P_{m_0}\}$.  

\vspace{2mm}  
\noindent Consider first the case where $H_{j^*}$ is selected (with probability $\tfrac{w_{j^*}}{\sum_{k=1}^{m_0} w_k}$). Then the probability of erroneously rejecting it under WAP is  
\[
\Pr\!\left(P_{j^*} \leq \tfrac{w_{j^*}}{\sum_{k=1}^{m_0} w_k}\alpha, \; P_{j^*} \leq P_i \;\;\forall P_i \in \mathbb{P}_0 \setminus \{P_{j^*}\}\right)  
= \tfrac{w_{j^*}}{\sum_{k=1}^{m_0} w_k}\,\alpha.
\]  

\noindent Next, suppose that some other $H_i$, $i \in I_0 \setminus \{j^*\}$, is selected. Then  
\[
\Pr\!\left(P_i \leq \tfrac{w_i}{\sum_{k=1}^{m_0} w_k}\alpha, \; P_i \leq P_l \;\;\forall P_l \in \mathbb{P}_0 \setminus \{P_i\}\right)  
= \tfrac{w_i}{\sum_{k=1}^{m_0} w_k}\,\alpha,
\]  
provided that the WAP rejection threshold for $H_i$ does not exceed the upper bound associated with $P_{j^*}$, i.e.,  
\[
\tfrac{w_i}{\sum_{k=1}^{m_0} w_k}\alpha \;\leq\; \tfrac{w_{j^*}}{\sum_{k=1}^{m_0} w_k} 
\quad \iff \quad \tfrac{w_{j^*}}{w_i} \geq \alpha.
\]  

\vspace{2mm}  
\noindent Thus, the condition $\tfrac{w_{j^*}}{w_i} \geq \alpha$ for all $i \in I_0 \setminus \{j^*\}$ guarantees that the probability of a Type~I error is bounded appropriately. Consequently,  
\[
\text{FWER} \;=\; \sum_{i=1}^{m_0} \Pr\!\left(P_i \leq \tfrac{w_i}{\sum_{k=1}^{m_0} w_k}\alpha, \; P_i \leq P_l \;\;\forall P_l \in \mathbb{P}_0 \setminus \{P_i\}\right) \;=\; \alpha.
\]  

\vspace{2mm}  
\noindent Therefore, whenever  
\[
\frac{\min_{i} w_i}{\max_{i} w_i} \,\geq\, \alpha,
\]  
the WAP attains equality $\text{FWER} = \alpha$ and is thus optimal in the sense that no critical value can be increased without violating FWER control. \hfill$\square$  

\vspace{2mm}  
\noindent \textit{Intuition.} The smallest weight $w_{j^*}$ acts as a bottleneck: if it is not too small relative to the largest weight, then every rejection threshold remains consistent with the uniform distribution of true null $p$-values. This balance ensures that the FWER is pushed exactly to the boundary $\alpha$, establishing the optimality of WAP.

\section*{Acknowledgements}
The authors are grateful to the Editor, the Associate Editor, and the reviewers for their careful evaluation and constructive comments, which have helped improve the clarity and presentation of the manuscript. The authors also thank Sanat Sarkar for his helpful discussions and valuable feedback.

\bigskip

\bibliographystyle{plainnat}
\bibliography{ref}   

\end{document}